\documentclass[aps,pra,twocolumn,floatfix,superscriptaddress]{revtex4-1}

\usepackage[utf8]{inputenc}
\usepackage[T1]{fontenc}
\usepackage{amsmath, amsthm, amssymb}
\usepackage{graphicx}
\usepackage{dcolumn,times}
\usepackage{bm,bbm}
\usepackage{color}

\definecolor{myurlcolor}{rgb}{0,0,0.7}
\definecolor{myrefcolor}{rgb}{0.8,0,0}
\usepackage[unicode=true,pdfusetitle, bookmarks=true,bookmarksnumbered=false,bookmarksopen=false, breaklinks=false,pdfborder={0 0 0},backref=false,colorlinks=true, linkcolor=myrefcolor,citecolor=myurlcolor,urlcolor=myurlcolor]{hyperref}

\newcommand{\proj}[1]{\ket{#1}\!\bra{#1}}
\newcommand{\Tr}[0]{\mathrm{Tr}}

\newcommand{\eref}[1]{(\ref{#1})}
\newcommand{\eqnref}[1]{Eq.~\eref{#1}}
\newcommand{\eqnsref}[2]{Eqs.~\eref{#1} and \eref{#2}}

\newcommand{\figref}[1]{Fig.~\ref{#1}}

\newcommand{\appref}[1]{App.~\ref{#1}}

\newcommand{\refcite}[1]{Ref.~\citep{#1}}
\newcommand{\ot}[0]{\otimes}

\newcommand{\ket}[1]{| #1 \rangle}

\newcommand{\bra}[1]{\langle #1 |}
\newcommand{\braket}[2]{\langle #1|#2\rangle}
\renewcommand{\t}[1]{\textrm{#1}}

\newtheorem{thm}{Theorem}
\newtheorem{cor}[thm]{Corollary}
\newtheorem{lem}[thm]{Lemma}

\DeclareGraphicsExtensions{.png,.pdf,.eps}
\graphicspath{ {./figs/} }

\begin{document}

\title{Asymptotic role of entanglement in quantum metrology}

\author{R. Augusiak}
\affiliation{ICFO--Institut de Ci\`encies Fot\`oniques, The Barcelona Institute of Science and Technology, 08860 Castelldefels (Barcelona), Spain}
\affiliation{Center for Theoretical Physics, Polish Academy of Sciences, Al.~Lotnik\'ow 32/46, 02-668 Warsaw, Poland}

\author{J. Ko\l{}ody\'{n}ski}
\affiliation{ICFO--Institut de Ciencies Fotoniques, The Barcelona Institute of Science and Technology, 08860 Castelldefels (Barcelona), Spain}

\author{A. Streltsov}
\affiliation{ICFO--Institut de Ciencies Fotoniques, The Barcelona Institute of Science and Technology, 08860 Castelldefels (Barcelona), Spain}

\author{M. N. Bera}
\affiliation{ICFO--Institut de Ciencies Fotoniques, The Barcelona Institute of Science and Technology, 08860 Castelldefels (Barcelona), Spain}
\affiliation{Harish-Chandra Research Institute, Chhatnag Road, Jhunsi, Allahabad 211 019, India}

\author{A. Ac\'in}
\affiliation{ICFO--Institut de Ciencies Fotoniques, The Barcelona Institute of Science and Technology, 08860 Castelldefels (Barcelona), Spain}
\affiliation{ICREA--Institucio Catalana de Recerca i Estudis Avan\c cats, Lluis Companys 23, 08010 Barcelona, Spain}

\author{M. Lewenstein}
\affiliation{ICFO--Institut de Ciencies Fotoniques, The Barcelona Institute of Science and Technology, 08860 Castelldefels (Barcelona), Spain}
\affiliation{ICREA--Institucio Catalana de Recerca i Estudis Avan\c cats, Lluis Companys 23, 08010 Barcelona, Spain}


\begin{abstract}
Quantum systems allow one to sense physical parameters beyond the reach of classical statistics---with resolutions
greater than $1/N$, where $N$ is the number of constituent particles independently probing a parameter.
In the canonical phase sensing scenario the \emph{Heisenberg Limit} $1/N^{2}$ may be reached, which requires,
as we show, both the \emph{relative size of the largest entangled block} and the 
\emph{geometric measure of entanglement} to be nonvanishing as $N\!\to\!\infty$.  Yet, we also demonstrate that 
in the asymptotic $N$ limit any precision scaling arbitrarily close to the Heisenberg Limit ($1/N^{2-\varepsilon}$ with 
any $\varepsilon\!>\!0$) may be attained, even though the system gradually becomes noisier and separable, so that 
both the above entanglement quantifiers asymptotically vanish. Our work shows that sufficiently large quantum systems 
achieve nearly optimal resolutions despite their relative amount of entanglement being arbitrarily small. In deriving our results, 
we establish the \emph{continuity relation of the quantum Fisher information} evaluated for a phaselike parameter,
which lets us link it directly to the geometry of quantum states, and hence naturally to the geometric measure
of entanglement.
\end{abstract}

\maketitle

\section{Introduction}
%
Quantum metrology is a vivid topic of research both at the theoretical and experimental 
levels \citep{Giovannetti2011,Dowling2014,Toth2014,*Demkowicz2015}.
With the help of quantum systems consisting of particles that independently sense a
parameter of interest one may attain sensing resolutions beyond the reach of classical 
statistics---beyond the so-called \emph{Standard Quantum Limit} (SQL) \citep{Giovannetti2004}.
This limit states that the Mean-Squared Error (MSE) of estimation may at best scale inversely 
to the number of particles employed, i.e., as $1/N$. Quantum mechanics allows one to beat the SQL 
and in the canonical phase-sensing scenario reach a $1/N^{2}$ resolution---the \emph{Heisenberg Limit} (HL)---%
a quantum enhancement of precision that limitlessly improves with $N$ \citep{Giovannetti2006}.
Spectacularly, quantum metrology schemes have been experimentally demonstrated to allow for 
enhanced sensing of phaselike parameters in optical interferometry \citep{Mitchell2004,*Nagata2007}, 
 e.g., in gravitational-wave detection \cite{LIGO2011,*LIGO2013}, but also 
in atomic-ensemble experiments of spectroscopy \citep{Leibfried2004,*Roos2006} and 
magnetometry \citep{Wasilewski2010,*Sewell2012}, as well as
in atomic clocks \citep{Appel2009,*Louchet2010}.

Quantum enhancement in metrology is only possible
thanks to the \emph{interparticle entanglement} 
exhibited by the quantum system employed \citep{Pezze2009}.
In fact, resolutions beyond classical limits have been used to 
prove the existence of large-scale entanglement in real atomic systems 
\citep{Riedel2010,*Krischek2011,*Strobel2014,*Lucke2014}.
The main obstacle in such experiments is the noise which destroys 
the entanglement and impairs the sensitivity \citep{Huelga1997,*Maccone2011}. 
For the attained precision to preserve the super-classical scaling,
the number of entangled particles, or, formally, the
entanglement \emph{producibility} \citep{Guhne2005,*Guhne2006} (or \emph{depth} \citep{Sorensen2001}) must grow with 
the system size \citep{Toth2012,*Hyllus2012}.
On the other hand, by studying the ultimate resolutions attainable with noisy
quantum systems, it has been shown that generic uncorrelated (independently disturbing the 
particles) noise-types limit the quantum enhancement to a constant factor 
\citep{Fujiwara2008,*Escher2011,*Demkowicz2012,*Knysh2014}. 
In terms of entanglement properties, such SQL-like sensitivities can then 
be reached for arbitrary large $N$ by grouping the constituent 
particles into separate entangled blocks of finite size \cite{Jarzyna2014}.
Although the protection of entanglement is thus of the highest priority 
for the super-classical precision scaling to be preserved, there also exist states 
that possess all their particles (genuinely) entangled but nonetheless are
useless for metrology \citep{Hyllus2010a}. Moreover, the nature 
of entanglement that is essential for metrological purposes remains unclear, 
as by employing large-scale but yet undistillable entanglement (which could 
be considered of the weakest type \citep{Horodecki2009,*Guhne2009}) one may 
still attain the HL resolution in phase sensing \citep{Czekaj2015}.

\begin{figure}[!t]
\includegraphics[width=1\columnwidth]{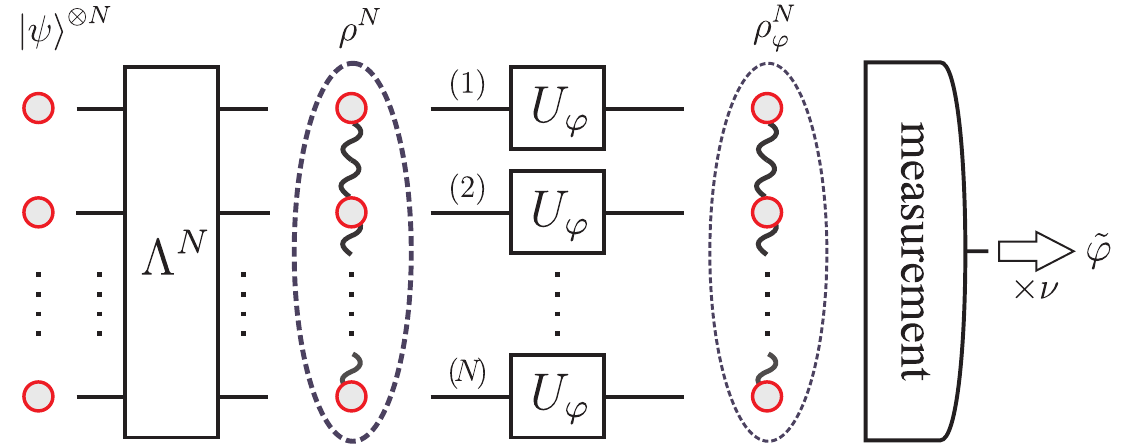}
\caption{%
\textbf{Quantum phase sensing  protocol}---designed 
to most precisely sense fluctuations of a phase-like parameter $\varphi$ 
around its given value. The system is prepared in an $N$-particle entangled state 
$\rho^N$ obtained from $\ket{\psi}^{\otimes N}$ by the preparation map $\Lambda^N$ that 
also incorporates noise. $\varphi$ is encoded on each particle by a unitary $U_\varphi$ and 
the final state $\rho_\varphi^N$ is measured. The procedure is repeated sufficiently many
 times ($\nu\!\gg\!1$) to construct most accurate parameter estimate $\tilde\varphi$.}
\label{fig:scheme}
\end{figure}

In this work, we connect the key metrological performance 
quantifier---the asymptotic scaling of precision---with the entanglement properties 
as quantified relatively to the size of the system employed.
To this end, we first establish a \emph{continuity relation for the quantum Fisher information} (QFI), 
which allows us to connect the metrological properties of quantum states to their geometry. Thanks to the derived continuity,
we are able to upper bound the QFI by the geometric measure of entanglement (GME) \citep{Shimony1995, *Wei2003}. As a result, we demonstrate that, although to attain
the \emph{exact} HL both the relative size of the largest entangled block
and the GME must be asymptotically nonvanishing, any precision scaling 
\emph{arbitrarily close} to HL, $1/N^{2-\varepsilon}$ with $\varepsilon\!>\!0$, is 
achievable despite both these entanglement quantifiers 
decaying with $N\!\to\!\infty$.

\section{Preliminaries}
\subsection{Metrology protocol}
We consider the noisy phase sensing protocol depicted in \figref{fig:scheme},
which allows us to unambiguously approach the problem. We encode
the parameter unitarily, so that the asymptotic precision scaling can be 
actually improved \cite{Fujiwara2008,*Escher2011,*Demkowicz2012,*Knysh2014},
and independently on each particle---so that the quantum-enhanced scaling is 
firmly constrained between SQL and HL ($1/N$ and $1/N^2$) and attributed
solely to the entanglement properties of the quantum state of the system \cite{Boixo2007,*Zwierz2010}. 
On the other hand, for any entanglement quantifier to be comparable 
with the asymptotic precision scaling it must be ``scale-independent'', i.e., 
it cannot grow with $N$ when considering a sequence of states of the same type 
\footnote{%
E.g., as (separable) GHZ states of $N$ particles attain a fixed resolution 
($1/N$) $1/N^2$, irrespectively of $N$ they should yield
the same metrologically relevant entanglement quantifier. 
}. Hence, a notion of entanglement ``size'' may only be quantified relatively 
to the total system size, while the entanglement ``amount'' must not change 
by simply increasing $N$. We define adequately both such notions below,
but let us already stress that the latter we find to be naturally emergent 
by relating the metrological properties of quantum states to their geometry.

We follow the frequentist approach to estimation which applies in the 
regime of sufficiently many independent protocol repetitions 
($\nu\!\gg\!1$ in \figref{fig:scheme}),
while sensing small parameter fluctuations around its certain known 
value \cite{Lehmann1998}. Then, the MSE, $\Delta^2\tilde \varphi$, of any (consistent and unbiased) estimator,
$\tilde \varphi$, of the parameter is ultimately lower limited by the 
\emph{Quantum Cram\'{e}r-Rao Bound} \citep{Braunstein1994,*Barndorff2000}:
\begin{equation}
\Delta^{2}\tilde\varphi\ge\frac{1}{\nu F_\t{Q}\!\left[\rho_{\varphi}^{N}\right]},
\,\t{where}\,\,
F_\t{Q}\!\left[\rho_{\varphi}^{N}\right]:=\textrm{Tr}\!\left\{\!\rho_{\varphi}^{N}L\!\left[\rho_{\varphi}^{N}\right]^{2}\right\}
\label{eq:QCRB}
\end{equation}
is the \emph{Quantum Fisher Information} (QFI) for a
given $N$-particle state $\rho_{\varphi}^{N}$
with $\varphi$ standing for the true parameter value. $L[\rho_{\varphi}^{N}]$
is the symmetric logarithmic derivative operator unambiguously 
defined via $\partial_{\varphi}\rho_{\varphi}^{N}\!=\!(L[\rho_{\varphi}^{N}]\rho_{\varphi}^{N}\!+\!\rho_{\varphi}^{N}L[\rho_{\varphi}^{N}])/2$
\citep{Helstrom1976}.

In the customary phase sensing protocol of \figref{fig:scheme} the estimated parameter is encoded 
onto the system state $\rho^{N}$ via $\rho_{\varphi}^{N}\!=\!U_{\varphi}^{\otimes N}\rho^{N}U_{\varphi}^{\dagger\otimes N}$
with $U_{\varphi}\!=\!\textrm{e}^{-\textrm{i}h\varphi}$ and $h$ being some fixed single-particle Hamiltonian.
Without loss of generality we assume the operator norm of $h$ to fulfil $\|h\|\!\le\!1/2$, so that the single-particle QFI 
generally satisfies $F_\t{Q}[\rho_{\varphi}^{1}]\!\le\!1$. As the parameter encoding is unitary, in what follows we 
may write the QFI for given $h$ as $F_\t{Q}[\rho^{N}]\!:=\!F_\t{Q}[\rho_{\varphi}^{N}]$ 
manifesting its independence of $\varphi$ \citep{Toth2014,*Demkowicz2015}.
Moreover, as the QFI is additive and convex \citep{Toth2014,*Demkowicz2015}, 
it must then fulfil $F_\t{Q}[\rho^N_\t{sep}]\!\le\!N$ for any separable $\rho^N_\t{sep}$. 
Hence, this proves that the SQL can be surpassed indeed \emph{only} when the quantum state 
$\rho^{N}$ exhibits entanglement \citep{Pezze2009}.

\subsection{Entanglement quantifiers}
In order to quantify the \emph{relative size} of entanglement contained in a given $\rho^N$
of \figref{fig:scheme}, we use the notion of \emph{producibility} \citep{Guhne2005,*Guhne2006}
(also termed \emph{entanglement depth} \citep{Sorensen2001}).
An $N$-particle pure state is termed \emph{$k$-producible} if it can 
be written as $\ket{\psi^{N}}=\otimes_{m=1}^{M\le N}\ket{\psi_{m}}$
with each $\ket{\psi_{m}}$ consisting of \emph{at most} $k$ particles.
This directly extends to mixed states:~a mixed state
$\rho^N$ is $k$-producible if it is a convex combination of pure 
$k$-producible states \citep{Guhne2005,*Guhne2006}.
Hence, for an $N$-particle Hilbert space $\mathcal{H}^{\otimes N}$,
the \emph{convex} sets of all $k$-producible states, which
we denote by $\mathcal{S}_k^N$, form a hierarchy
[$\mathcal{S}^N_1\!\subset\!\mathcal{S}^N_2\!\subset\!\dots\!\subset\!\mathcal{S}^N_N$]
that we schematically depict in \figref{fig:sets}. Note that $\mathcal{S}^N_1$ is just the set of fully separable states, 
$\mathcal{S}^N_N$ is the set of all states acting on $\mathcal{H}^{\otimes N}$,
while $\mathcal{S}^N_N\!\setminus\!\mathcal{S}^N_{N-1}$ contains ones that are 
genuinely entangled---they do not admit any form 
of separability. Crucially, the concept of producibility allows us to 
define for any $N$-particle state $\rho^{N}_{(l)}$ that is 
$l$-producible but \emph{not} $(l\!-\!1)$-producible, i.e., 
$\rho^{N}_{(l)}\in\mathcal{S}^N_{(l)}\!:=\!\mathcal{S}^N_{l}\!\setminus\!\mathcal{S}^N_{l-1}$,
the \emph{relative size of Largest Entangled Block (LEB)} of particles as $R_\t{LEB}\!:=\!l/N$.
Thus, $R_\t{LEB}$ is the ratio of the size of the largest subgroup of particles
that are entangled to the total particle number.
When describing the precision scaling attained by
metrology protocols, one deals with 
the $N\!\to\!\infty$ limit. Note that in such an asymptotic regime
the LEB may be divergent 
despite $R_\t{LEB}$ vanishing with $N$. 
Hence, if one was to associate the entanglement size with the number of 
particles being entangled via LEB, in many situations it would be infinite for
$N\!\to\!\infty$. In contrast, $R_\t{LEB}$ adequately then takes 
values within the interval $[0,1]$ depending on the sequence 
of states considered.

\begin{figure}[!t]
\includegraphics[width=0.9\columnwidth]{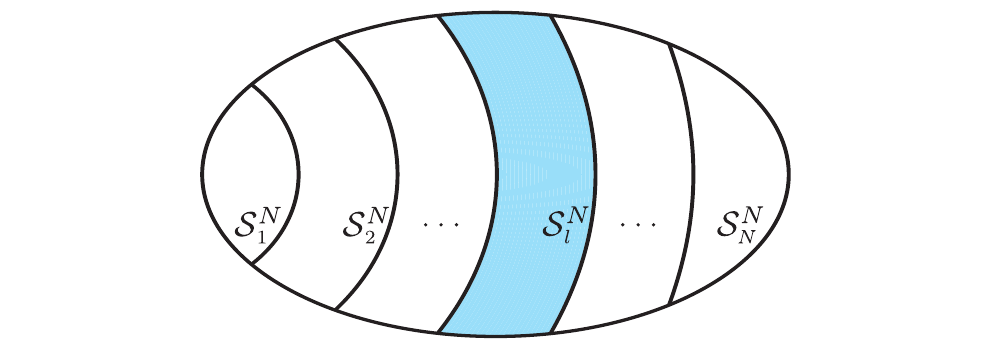}
\caption{%
\textbf{Hierarchy of convex sets $\mathcal{S}_k^N$ containing all $k$-producibile states.}
The set $\mathcal{S}_N^N$ contains all states acting on $\mathcal{H}^{\otimes N}$, while $\mathcal{S}_1^N$ is 
its subset of separable states. Shaded region is the set $\mathcal{S}_l^N\!\setminus\!\mathcal{S}_{l-1}^N$ of
states with the 
relative size of largest entangled block:~$R_\t{LEB}=l/N$.}
\label{fig:sets}
\end{figure}

On the other hand, in order to quantify the amount of entanglement exhibited by 
$\rho^N$ in \figref{fig:scheme}, we employ the \emph{geometric measure of entanglement} (GME)
that is defined for pure states as
$E_\t{G}[\ket{\psi^N}]\!:=\!1\!-\!\max_{\ket{\phi^N} \in \mathcal{S}_1^N}|\langle\phi^N|\psi^N\rangle|^2$
\citep{Shimony1995,*Wei2003}. Its definition naturally generalises to mixed states 
through the convex-roof construction \citep{Shimony1995,*Wei2003}:
\begin{equation}
E_\t{G}[\rho^N]
\;:=
\inf_{\{p_i,\ket{\psi_i^N}\}}\sum_i p_i\, E_\t{G}[\ket{\psi_i^N}]
\label{eq:E_G}
\end{equation}
with the infimum taken over all ensembles $\{p_i,\ket{\psi_i^N}\}$ such that $\rho^N\!=\!\sum_ip_i\proj{\psi_i^N}$.
However, one may show that definition \eref{eq:E_G} may be equivalently obtained by 
employing the Uhlmann fidelity, $F(\rho,\sigma)\!:=\!\Tr\sqrt{\sqrt\sigma\rho\sqrt\sigma}$ \citep{Bengtsson2006},
so that $E_{G}[\rho^{N}]\!=\!1\!-\!\max_{\sigma^{N}\in\mathcal{S}_{1}^{N}}F^2(\rho^{N},\sigma^{N})$
\cite{Streltsov2010}. Crucially, thanks to its geometrical formulation, the GME is independent of the 
particle number $N$. In particular, it effectively measures the distance to separable
states independently of the Hilbert space dimension. To see this, note 
that the GME obeys the following inequality:
\begin{equation}
1-\sqrt{1-E_{G}(\rho^{N})}
\leq
\min_{\sigma^{N}\in\mathcal{S}_{1}^{N}} T(\rho^{N},\sigma^{N})
\leq
\sqrt{E_{G}(\rho^{N})},\label{eq:faithful}
\end{equation}
which directly follows from the Fuchs-van de Graaf relation
between fidelity and trace-distance, $T(\rho,\sigma)\!:=\!\|\rho-\sigma\|_{1}/2$, of quantum states \cite{Fuchs1999}. 
In particular, for any sequence of states $\{\rho^{N}\}$
that is bounded away from the set of fully separable states, in the sense that
$\min_{\sigma^{N}\in\mathcal{S}_{1}^{N}} T(\rho^{N},\sigma^{N})\!\geq\!c$
with some constant $c\!>\!0$, the GME is also bounded away from zero
as $E_{G}(\rho^{N})\!\geq\!c^{2}$. On the other hand, the vanishing GME of a sequence $\{\rho^N\}$
implies that its elements must converge to the set of separable states as $N\!\to\!\infty$.
As a result, the GME may have been used to demonstrate, e.g., that typical states---%
by exhibiting high GME---have high entanglement \cite{Gross2009}.

\section{Results}
%
%
\subsection{Continuity of QFI}%
Our first result is the continuity relation for the QFI, which
will later allow us to naturally connect the GME of a state with its metrological properties. 
More precisely, exploiting the \emph{purifications-based definition} of QFI \citep{Fujiwara2008,Escher2011},
we upper bound in \appref{app:QFI_cont} the difference of QFIs for any two quantum states 
via their geometrical separation, in particular, via their \textit{fidelity}, \emph{trace} or \emph{Bures distance}. In the special case of one of the states being pure, we 
additionally tighten the corresponding bound utilising the \emph{convex--roof-based definition} of QFI \citep{Toth2013,*Yu2013}.
The result may be summarised by the following inequality holding for
any two $\rho^N\!,\sigma^N\!\in\!\mathcal{B}(\mathcal{H}^{\otimes N})$ (see \appref{app:QFI_cont} for the proof):
\begin{equation}
\left|F_\t{Q}\!\left[\rho^{N}\right]-F_\t{Q}\!\left[\sigma^{N}\right]\right|
\;\le\;
\xi\, \sqrt{1-F\!\left(\rho^N,\sigma^N\right)^2}\, N^{2},
\label{eq:cont}
\end{equation}
where $F(\rho,\sigma)$ is again the Uhlmann fidelity, while $\xi\!=\!8$ for general quantum 
states and $\xi\!=\!6$ if one of them is pure.

Let us first stress the general power of the continuity relation \eref{eq:cont} when 
used for comparing metrological properties between multipartite states.
It straightforwardly follows from \eqnref{eq:cont} that for any pair $\rho^N\!,\sigma^N$ (see also \appref{app:QFI_cont}):
\begin{equation}
\label{eq:cont2}
F_\t{Q}\!\left[\rho^{N}\right]
\;\le\; 
F_\t{Q}\!\left[\sigma^{N}\right]+
\xi\, \sqrt{2\;T(\rho^N,\sigma^N)}\, N^{2}
\end{equation}
with $T(\rho,\sigma)$ denoting the trace distance as before. 
Hence, \eqnref{eq:cont2} directly implies that given a sequence of states 
$\{\sigma^N\}$ that \emph{do not attain} the HL,
i.e., $F_\t{Q}[\sigma^N]/N^2\!\to\! 0$ with $N$, 
\emph{any} other sequence $\{\rho^N\}$ which consists 
of states those successively converge to $\{\sigma^N\}$, so 
that $T(\rho^N,\sigma^N)\!\to\!0$ with $N$, \emph{cannot} attain the HL either. 
In particular, recalling that separable states do not allow for any quantum-enhanced sensitivity, 
no sequence of states tending to the set of fully separable states $\mathcal{S}_1^N$ may 
attain the HL. On the other hand, taking in contrast $\{\rho^N\}$ in \eqnref{eq:cont2} 
as the reference sequence that \emph{attains} the HL, i.e.,
$F_\t{Q}[\rho^N]/N^2\!\to\!c\!>\!0$  with $N$,
\eqnref{eq:cont2} proves that \emph{any} other sequence 
$\{\sigma^N\}$ \emph{must also} attain the HL, as long as 
$\xi \sqrt{2\,T(\rho^N,\sigma^N)}\to c'\!<\!c$ while $N\!\to\!\infty$.
For instance, choosing as reference the optimal GHZ states, which yield
$F_\t{Q}[\psi_\t{GHZ}^N]\!=\!N^2$ and, hence, $c\!=\!1$ ($\xi\!=\!6$), 
we see that any other sequence of states $\{\sigma^N\}$ that maintain 
their $T(\psi_\t{GHZ}^N,\sigma^N)\!<\!1/72$ with $N$ must also 
follow the HL. This is consistent with recent profound methods which
focus on Dicke-state sequences and imply
$T(\psi_\t{GHZ}^N,\sigma^N)\!<\!\sqrt{3}/2$ to be sufficient 
in case of GHZ states \cite{Apellaniz2015}.

Surprisingly, \eqnref{eq:cont2} opens an interesting possibility:~there 
may exist two sequences of states which asymptotically converge
despite contrasting metrological properties.
Consider a sequence of states $\{ \sigma^{N}\}$ such that
$F_\t{Q} [\sigma^{N}]\!\sim\!N^{\alpha}$ for sufficiently large $N$,
yielding a $1/N^{\alpha}$ asymptotic resolution with $0\!<\!\alpha\!<\!2$ 
(possibly even sub-SQL). \eqnref{eq:cont2} does not exclude
the existence of another sequence $\{\rho^{N}\}$
with $T(\rho^N,\sigma^N)\!\to\! 0$ for $N\!\to\! \infty$
that nonetheless attains \emph{any} improved precision scaling
$1/N^{2-\varepsilon}$ with $0\!<\!\varepsilon\!<\!2\!-\!\alpha$.
All what \eqnref{eq:cont2} imposes is that
$\{\rho^{N}\}$ approaches $\{\sigma^{N}\}$ slow enough, so that
$T(\rho^N,\sigma^N)\!\gtrsim\!1/N^{2\varepsilon}$ as $N\!\to\!\infty$.
In the context of entanglement, there may thus exist sequences 
approaching the set of \emph{separable} states 
 but preserving precision scaling arbitrarily close to HL.
We later provide examples of such sequences.

\subsection{Relating QFI to entanglement}%
We first recall the result of
\cite{Toth2012,*Hyllus2012} relating the notions of QFI and $k$-producibility:~for 
any $k$-producible state $\sigma^{N}\!\in\!\mathcal{S}_k^N$, the QFI is upper bounded as
\begin{equation}
F_\t{Q}[\sigma^{N}]\;\le\;\left\lfloor \frac{N}{k}\right\rfloor k^{2}+\left(N-\left\lfloor \frac{N}{k}\right\rfloor k\right)^{2}\le \;k\,N,
\label{eq:k-prod_bound}
\end{equation}
where $\lfloor x\rfloor\!:=\!\t{floor}[x]$. The above bound importantly implies 
that for states with fixed producibility $k$  (independent of $N$) the quantum 
enhancement is limited to a constant factor \cite{Toth2012,*Hyllus2012}. 
Hence, for a super-classical precision scaling to be possible the 
preparation map $\Lambda^N$ in \figref{fig:scheme} must output states such 
that their producibility constantly rises with increasing $N$. 

On the other hand, in terms of $R_\t{LEB}$, \eqnref{eq:k-prod_bound} equivalently 
reads:~$F_\t{Q}[\sigma^{N}]\le R_\t{LEB}N^2$. Thus, the \emph{exact} HL can be 
attained \emph{only if} $R_\t{LEB}$ does not vanish in the asymptotic $N$ limit, what 
requires the relative size of entanglement to be maintained with increasing $N$. However, 
similarly to the continuity relation \eref{eq:cont2}, \eqnref{eq:k-prod_bound} leaves open the 
existence of sequences attaining scalings \emph{arbitrarily close} to HL despite their $R_\t{LEB}$ 
tending to zero with $N$ (it requires the size of the particle LEB to grow as $N^{1-\varepsilon}$, letting 
$R_\t{LEB}$ vanish as $N^{-\varepsilon}$). Operationally, in order to reach a super-classical scaling, 
it is thus enough for the $\Lambda^N$ of \figref{fig:scheme} to prepare states with the effective 
number of entangled particles rising with $N$, yet at such a rate that its ratio to the total particle 
number is constantly decreasing. One may thus argue that the preparation map $\Lambda^N$ 
of \figref{fig:scheme} is then experimentally easier to implement, as it does not require the relative 
size of entanglement to be maintained with increasing $N$ (e.g., while squeezing an atomic ensemble
\cite{Kitagawa1993}), especially when dealing with systems of macroscopic size
\footnote{%
E.g., in optical interferometry with strong laser beams $N\!\approx\!10^{12}/\t{ns}$
\citep{LIGO2011,*LIGO2013}, while in atomic-ensemble experiments $N\!\gtrapprox\!10^5$
\citep{Wasilewski2010,*Sewell2012,Appel2009,*Louchet2010}.
}.

Let us now provide the second main result relating
the QFI and the GME.
To this end, we show that \eqnref{eq:cont} (with $\xi\!=\!6$)
may be utilised to upper bound the QFI as (see \appref{app:geom_meas} for the proof):
\begin{equation}
F_\t{Q}\!\left[\rho^{N}\right]\;\le\; N+6\,\sqrt{E_\t{G}\!\left[\rho^{N}\right]}\,N^{2}.
\label{eq:geom_meas_bound}
\end{equation}
As an aside, note that the formula \eref{eq:geom_meas_bound} may be straightforwardly 
generalised to any $k\!>\!1$ with help of the bound \eref{eq:k-prod_bound}
and by defining the \textit{geometric measure of $k$-producibility} after replacing
$\mathcal{S}_1^N$ with $\mathcal{S}_k^N$ in \eqnref{eq:E_G} (see \appref{app:geom_meas}).

Inequality \eref{eq:geom_meas_bound} implies that the \emph{exact} HL
can \emph{only} be attained if the entanglement is asymptotically nonvanishing, i.e.,
any sequence $\{\rho^N\}$ with GME vanishing for $N\!\to\!\infty$ \emph{cannot} reach the 
$1/N^2$ scaling. Still, \eqnref{eq:geom_meas_bound} does not exclude the possibility that 
any resolution \emph{arbitrarily close} to HL is attained by a sequence 
$\{\rho^N\}$, whose elements exhibit vanishingly small geometric measure of entanglement as $N\to\infty$. 
In particular, \eqnref{eq:geom_meas_bound} just requires the GME to decay 
slowly enough, so that as long as asymptotically $E_\t{G}[\rho^N]\gtrsim1/N^{2\varepsilon}$,
any resolution $1/N^{2-\varepsilon}$ is allowed.

\subsection{Almost the HL with vanishing $R_\t{LEB}$ and GME}%
In order to affirm the above claims, we now provide examples of state sequences---consisting 
of either pure or mixed states---that attain precision scalings arbitrarily close to HL 
despite their relative size and amount of entanglement, as quantified by $R_\t{LEB}$ and GME respectively, 
vanishing with $N\!\to\!\infty$. We return to the phase sensing scenario of \figref{fig:scheme} with 
the parameter $\varphi$ being unitarily encoded via the single-particle Hamiltonian $h\!=\!\sigma_z/2$.

First, let us consider non-maximally entangled states:
\begin{equation}
\ket{\psi^N_p}\;:=\;\sqrt{p}\ket{0}^{\ot N}+\sqrt{1-p}\ket{1}^{\ot N}
\label{eq:non-max_ent}
\end{equation}
with $0\!\le\!p\!\le\!1/2$, so that $\ket{\psi^N_{1/2}}$ is 
the GHZ state of $N$ qubits. The metrological capabilities
of states \eref{eq:non-max_ent} were studied in \refcite{Hyllus2010a},
where it was shown that by making $p$ vanish quickly enough with $N$,
states \eref{eq:non-max_ent} do not surpass SQL despite being genuinely 
entangled for any $p\!>\!0$. On the contrary, we focus on the fact that 
states \eref{eq:non-max_ent} also allow for resolutions arbitrarily close 
to HL even when $p\!\to\!0$ as $N\!\to\!\infty$. However, in order to also control and 
vary their LEB, we tailor them to $\ket{\psi^N_{p,l}}\!:=\!\ket{\psi^l_p}\otimes\ket{0}^{\otimes N-l}$,
so that their $R_\t{LEB}\!=\!l/N$ for any $p\!>\!0$.
As $E_\t{G}[\ket{\psi_{p,l}^N}]\!=\!p$ \citep{Wei2003}, we may then
rewrite their QFI as
$F_\t{Q}[\ket{\psi_{p,l}^N}]\!=\!4p(1-p)l^2\!=\!4 E_\t{G}(1-E_\t{G})R_\t{LEB}^2 N^2$.
Thus, by setting both $E_\t{G}=1/N^{\varepsilon_1}$ and $R_\t{LEB}=1/N^{\varepsilon_2}$
to vanish with $N$ for any $\varepsilon_1,\varepsilon_2\!>\!0$, we obtain the QFI 
to scale as $F_\t{Q}[\ket{\psi^N_{p,l}}]\sim N^{2-\varepsilon_1-2\varepsilon_2}$, 
which in turn yields the desired arbitrariy close to HL resolution $1/N^{2-\varepsilon_1-2\varepsilon_2}$.

Now, let us turn to the case of mixed states and consider
$N$-qubit Werner-type states \cite{Werner1989}:
\begin{equation}
\rho^N_p=p\proj{\psi_{1/2}^N}+(1-p)\frac{\mathbbm{1}_{2^N}}{2^N}.
\label{eq:GHZmixture}
\end{equation}
The QFI of $\rho^N_p$ reads $F_\t{Q}[\rho_p^N]\!=\!N^2p^2/[p\!+\!(1\!-\!p)/2^{N-1}]$ \cite{Toth2014},
and for sufficiently large $N$ simplifies to $p N^2$ (independently whether $p$ depends on $N$). 
Although the GME can be exactly evaluated for these states \cite{GHZStates},
for our purposes it is enough to use the upper bound $E_\t{G}[\rho_p^N]\leq p/2$, which stems
from the convexity of GME and may be shown to be saturated for
$N\!\to\!\infty$ (see \appref{app:EG_est}). Thus, for sufficiently large $N$ we may
write  $F_\t{Q}[\rho_p^N]\!\ge\! 2 E_\t{G} N^2$.
Note that by setting $p=1/N^{\varepsilon}$, leading to $E_\t{G}[\rho_p^N]\!\le\! 1/(2N^{\varepsilon})$, we
actually let the white noise increase with $N$, so that
the state \eref{eq:GHZmixture} becomes fully depolarised in the asymptotic $N$ limit. 
Nevertheless, 
the QFI scales then at least as $F_\t{Q}[\rho_p^N]\gtrsim N^{2-\varepsilon}$
leading to the claimed $1/N^{2-\varepsilon}$ resolution.
Moreover, it has been proven that for states \eref{eq:GHZmixture} to be genuinely entangled 
$p\!>\!(2^{N-1}\!-\!1)/(2^{N}\!-\!1)$ \cite{Guhne2010}, which for large $N$ converges to $1/2$ from 
above. Hence, by letting $p\!\to\!0$ as $N\!\to\!\infty$ we obtain a sequence of states 
that quickly seize to be genuinely entangled with strictly $R_\t{LEB}\!<\!1$. However, in order 
to prove that $R_\t{LEB}$ can be made vanishing, similar to the pure states
case, we tailor the states \eref{eq:GHZmixture} accordingly to
\begin{equation}
\rho_{p,l}^N\;:=\;p\proj{\psi_{1/2,l}^N}+(1-p)\frac{\mathbbm{1}_{2^N}}{2^N},
\label{eq:GHZmixture_l}
\end{equation}
so that $R_\t{LEB}\!\le\!l/N$ may be assured. Following 
the same argumentation as in the case of \eqnref{eq:GHZmixture}
\cite{Toth2014}, the QFI of states \eref{eq:GHZmixture_l} can then be shown to 
simplify to $F_\t{Q}[\rho_{p,l}^N]\!\approx\!p l^2$ for sufficiently large $N$.
Hence, as the GME has to still obey $E_\t{G}[\rho_{p,l}^N]\!\leq\! p/2$,
the QFI of states \eref{eq:GHZmixture_l} must asymptotically scale at least as
$F_\t{Q}[\rho_{p,l}^N]\!\gtrsim\!2 E_\t{G} R_\t{LEB}^2 N^2$.
Thus, as desired, also the mixed states  \eref{eq:GHZmixture_l} allow us to set
both $E_\t{G}\!=\!1/N^{\varepsilon_1}$ and $R_\t{LEB}\!=\!1/N^{\varepsilon_2}$
vanishing, but still attain the $1/N^{2-\varepsilon_1-2\varepsilon_2}$ resolution
despite becoming completely depolarised in the asymptotic $N$ limit.

Although the above pure- and mixed-state sequences 
demonstrate that, indeed, both the GME and $R_\t{LEB}$
may be set vanishing as $N\!\to\!\infty$, while maintaining the arbitrarily close to HL
resolutions, the exemplary sequences do not asymptotically saturate 
the bounds on the QFI set by \eqnsref{eq:k-prod_bound}{eq:geom_meas_bound}.
In the latter case, we expect this to be a consequence of the QFI continuity
relation \eref{eq:cont} actually not being asymptotically saturable due
to the ``square-root'' dependence on the distance between quantum states
appearing in its form. 

To put our results on firm ground, let us assume that one wants to
attain a super-classical resolution that is close to HL, e.g., $1/N^{1.7}$. Then,
in the case of pure \eref{eq:non-max_ent} and mixed \eref{eq:GHZmixture_l} states 
it may be reached after letting both the $R_\t{LEB}$ and GME vanish with $\varepsilon_1\!=\!\varepsilon_2\!=\!0.1$.
Hence, when considering the large-particle-number regime of $N\!\approx\!10^6$ (typical to atomic-ensemble experiments 
\citep{Wasilewski2010,*Sewell2012,Appel2009,*Louchet2010}), one requires $E_\t{G}\!\approx\!0.25$, 
which is half the entanglement of the GHZ state, and $R_\t{LEB}\!\approx\!25\%$, that is, one-fourth 
of particles need to be entangled.

\begin{figure}[!t]
\includegraphics[width=1\columnwidth]{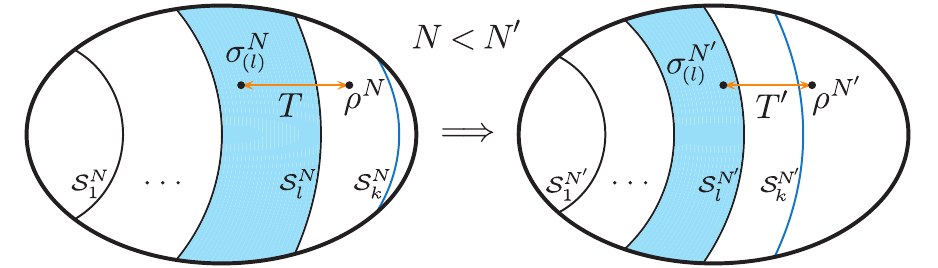}
\caption{%
\textbf{Collapse of the hierarchy of sets} $\mathcal{S}_k^N$
allowing two sequences $\{\sigma_{(l)}^N\}$ and $\{\rho^N\}$
to approach one another with $N$, despite their contrasting
metrological properties. States $\sigma_{(l)}^N$ are of constant $\t{LEB}\!=\!l$, which
(due to \eqnref{eq:k-prod_bound}) constrains their QFI to scale at most linearly with $N$:~ 
$F_{\mathrm{Q}}[\sigma_{(l)}^N]\!\le\!lN$. 
$\rho^N$ are states with their LEB rising with $N$, so that their QFI is taken to scale
as $N^{\alpha}$ with some $\alpha\!<\!2$. Still, it is possible to choose the states 
$\sigma_{(l)}^N$ and $\rho^N$ (e.g., ones of \eqnref{eq:GHZmixture_l}) so that for 
sufficiently large $N$ they become arbitrarily close to each other.  In particular, for the two 
$N'\!>\!N$ depicted above, both $\sigma_{(l)}^{N}\!\in\!\mathcal{S}_{(l)}^{N}$ and $\sigma_{(l)}^{N'}\!\in\!\mathcal{S}_{(l)}^{N'}$, whereas
$\rho^{N}\!\in\!\mathcal{S}_k^N$ but $\rho^{N'}\!\notin\!\mathcal{S}_k^{N'}$ for some $k>l$,
although the trace distance ($T'\!<\!T$) between $\sigma_{(l)}^N$ and $\rho^N$ decreases with $N$. 
This is possible because the sets $\mathcal{S}_l^N$ collapse faster than $\rho^N$
approach $\sigma_{(l)}^N$. 
}
\label{fig:seq_states}
\end{figure}

\subsection{Geometric interpretation of the results}%
In \figref{fig:seq_states}, we schematically present an exemplary path
that elements of sequences $\{\sigma_{(l)}^{N}\}$ and $\{\rho^{N}\}$ 
should take for the above described phenomenon to be possible:~%
despite becoming arbitrarily close to each other as $N\!\to\!\infty$, the states 
$\{\sigma_{(l)}^{N}\}$ and $\{\rho^{N}\}$ have drastically
different metrological properties. To be more precise, let the states 
$\sigma_{(l)}^{N}\!\in\!\mathcal{S}^{N}_{(l)}$ 
be of constant $\t{LEB}\!=\!l$ for all $N$. According to \eqnref{eq:k-prod_bound},
their QFI is thus always constrained by $lN$, so that they only may
yield an SQL-like precision scaling. On the other hand, 
let $\rho^N$ be states whose LEB grows with $N$ in a way
that they attain an asymptotic precision scaling arbitrarily close to HL.
Crucially, the two sequences exhibit highly contrasting metrological 
properties in the asymptotic $N$ limit. Still, it is possible to 
choose them in a way that the geometric distance between their 
consecutive elements gradually vanishes as $N\!\to\!\infty$.
As shown in \figref{fig:seq_states}, this is possible as the elements of 
$\{\rho^{N}\}$ are constantly ``overtaken'' by the 
boundaries of sets of higher producibility, while the hierarchy collapses with increasing $N$. 
We explicitly draw the boundaries of the $k$-producible sets for particle 
numbers $N\!<\!N'$, in order to emphasize that our results suggest rapid shrinkage of the sets with $N$.
In particular, note that in \figref{fig:seq_states}:~$\rho^N\!\in\!\mathcal{S}_k^N$ but $\rho^{N'}\!\notin\!\mathcal{S}_k^N$; 
even though $T'\!=\!T(\rho^{N'},\sigma_{(l)}^{N'})<T\!=\!T(\rho^{N},\sigma_{(l)}^{N})$.
We expect such a phenomenon to be the consequence of the volume of each $\mathcal{S}_k^N$ collapsing
exponentially with $N$, which, according to our best knowledge has only been proven 
for the set of separable states (i.e., for $l=1$) \citep{Gurvits2005,*Szarek2005}.

\section{Conclusion}%
We have studied restrictions that entanglement features 
impose on the asymptotic metrological performance of quantum states. 
First, by establishing the continuity relation for the QFI, we have related the 
metrological properties of states to their underlaying geometry. This allowed us to naturally link their 
metrological utility to their entanglement content as measured by the geometric measure of entanglement.
As a result, we have shown that for the HL to be attained in the asymptotic $N$ limit 
both the relative size and amount of entanglement (as quantified by $R_\t{LEB}$ 
and GME respectively) cannot vanish. For instance, the states that exhibit undistillable entanglement, 
but still attain the exact HL \cite{Czekaj2015}, must thus asymptotically possess finite 
$R_\t{LEB}$ and GME. On the contrary, we have demonstrated that any precision scaling 
arbitrarily close to HL may be reached even though both $R_\t{LEB}$ and GME 
vanish as $N\!\to\!\infty$. In the presence of global depolarisation, 
this still allows the decoherence strength to be increasing with $N$, which contrasts the case of 
uncorrelated noise-types whose strength must decrease with system size for a scaling 
quantum-enhancement to be observed \cite{Fujiwara2008,Escher2011,Demkowicz2012,Knysh2014}.
As uncorrelated noises yield $R_\t{LEB}\!\sim\!1/N$ in the asymptotic $N$ limit \cite{Jarzyna2014}, 
our results provide a new (entanglement-degrading) interpretation of their destructive impact. 
We hope that our work can thus be beneficial in proposing novel noise-robust metrology 
schemes that attain super-classical resolutions by employing quantum states with just the 
necessary entanglement properties. As the metrological usefulness of a quantum state is 
directly related to its \emph{macroscopicity} \cite{Froewis2012a,*Tichy2015}, all our 
results also apply in this context. Let us finally notice that our continuity of the QFI 
has recently been used by some of us to study the typical metrological properties of
various ensembles of quantum states \cite{Oszmaniec2016}.

A natural open question to ask is how the conclusions of our work
vary if one considers single-shot ($\nu\!=\!1$ in \figref{fig:scheme}) metrology protocols, in which the estimated 
parameter may not be assumed to be fluctuating around a known value
and the Bayesian approach to estimation must be pursued \cite{Berry2000,*Bagan2001b}.
Yet, we expect the requirements on the entanglement to be then much more stringent 
due to the lack of sufficient statistics.

\begin{acknowledgments}
%
We acknowledge enlightening discussions with Christian Gogolin, Marcin Jarzyna, 
and Paul Skrzypczyk. This work has been supported by the 
ERC AdG OSYRIS and CoG QITBOX, Axa Chair in Quantum Information Science,
John Templeton Foundation, EU IP SIQS, EU STREP QUIC and EQuaM, Spanish Ministry 
National Plans FOQUS (FIS2013-46768) and MINECO (Severo Ochoa Grant No.~SEV-2015-0522), 
Fundaci\'{o} Privada Cellex, Generalitat de Catalunya (Grant No.~SGR 874 and 875),
Foundation for Polish Science (START scholarship) and Alexander von Humboldt Foundation.
 J.K.~and R.A.~further acknowledge funding from the European Union’s Horizon 2020 research 
and innovation programme under the Marie Sk\l{}odowska-Curie Q-METAPP and N-MuQuaS 
Grants No.~655161 and 705109.
\end{acknowledgments}

\appendix

\section{Continuity of QFI}
\label{app:QFI_cont}

Here we present a detailed proof of the continuity relation for the
quantum Fisher information (QFI), which we extensively use in the main text.
Yet, in order to also establish a common notation and preliminary notions,
we firstly introduce the basic concepts of:~a \emph{purification of a mixed
state}, and that of the \emph{Uhlmann fidelity} and the \emph{Bures distance} \cite{Nielsen2000,Bengtsson2006}.

Let us consider a quantum system represented by a mixed state $\rho$
acting on a Hilbert space $\mathcal{H}_{S}=\mathbbm{C}^{d}$. It follows
that $\rho$ can always be represented by a pure state from a larger
Hilbert space. Concretely, there exists $\ket{\psi}\in\mathcal{H}_{S}\ot\mathcal{H}_{E}=\mathbbm{C}^{d}\ot\mathbbm{C}^{d'}$
with $d'=\mathrm{rank}(\rho)$ such that $\rho=\Tr_{E}\proj{\psi}$.
This representation is, however, not unique because any pure state
related to $\ket{\psi}$ via $\ket{\psi'}=\mathbbm{1}_{S}\ot V_{E}\ket{\psi}$
with $V_{E}$ being some partial isometry $(V_{E}^{\dagger}V_{E}=\mathbbm{1})$
is also a \emph{purification} of $\rho$. At this point it is important to
mention that any such $V_{E}$ can be extended to a unitary operation
by properly enlarging the ``environmental'' Hilbert space $\mathcal{H}_{E}$,
and so any two purifications of a given $\rho$ are thus related by a unitary
operation acting on $\mathcal{H}_{E}$ \cite{Nielsen2000}.

Then, the \emph{Uhlmann fidelity} of a pair of density matrices $\rho$ and
$\sigma$ acting on $\mathcal{H}_{S}=\mathbbm{C}^{d}$ is defined
through
\begin{equation}
F(\rho,\sigma):=\left\Vert \sqrt{\rho}\sqrt{\sigma}\right\Vert _{1}=\Tr\sqrt{\sigma^{1/2}\rho\,\sigma^{1/2}},
\label{eq:fidelity}
\end{equation}
where $\| \cdot\|_1$ stands for the trace norm defined as $\|X\|_1=\Tr\sqrt{X^{\dagger}X}$.
If at least one of these two states is pure, say $\sigma=\proj{\phi}$,
then the above formula simplifies to $F(\rho,\ket{\phi})=\sqrt{\langle\phi|\rho|\phi\rangle}$,
and if also $\rho=\proj{\psi}$ then
$F(\ket{\psi},\ket{\phi})=|\braket{\psi}{\phi}|$ is just the overlap of the two pure states.

For further benefits let us also mention that for any pair of mixed
states $\rho$ and $\sigma$ their fidelity can be expressed in terms
of the fidelity of their purifications, denoted by $\ket{\psi}$ and
$\ket{\phi}$, respectively. More precisely,
\begin{equation}
F(\rho,\sigma)=\max_{\ket{\psi}}F(\ket{\psi},\ket{\phi})=\max_{\ket{\psi}}|\langle\psi|\phi\rangle|
\label{eq:fidelityPur}
\end{equation}
where the maximization is performed over all purifications of $\rho$,
but equally well can be performed over the purifications of $\sigma$
\cite{Nielsen2000}. The Uhlmann fidelity does not fulfil properties
of a \emph{measure of distance} between quantum states \cite{Bengtsson2006},
yet with its help one may define the
so-called \textit{Bures distance} \cite{Bengtsson2006}:
%
\begin{equation}
D_{B}(\rho,\sigma)=\sqrt{2\left[1-F(\rho,\sigma)\right]}.
\label{eq:Bures_distance}
\end{equation}

Having these notions at hand, we can now pass to the continuity relations
of the QFI. Let us first recall that in our case the parameter $\varphi$
is encoded on a state with the aid of a unitary operation,
so that $\rho_\varphi=U_\varphi\rho\, U_\varphi^\dagger$
where $U_{\varphi}=\mathrm{e}^{-\mathrm{i} H \varphi}$ and $H$
is a given parameter-encoding Hamiltonian.
In such case, the QFI most generally reads
\begin{equation}
 F_\t{Q}[\rho;H]:=F_\t{Q}[\rho_\varphi]
=
2\sum_{k,l}\frac{(\lambda_{k}-\lambda_{l})^{2}}{\lambda_{k}+\lambda_{l}}|\langle\xi_{k}|H|\xi_{l}\rangle|^{2},
\label{eq:QFIUnit}
\end{equation}
where $\lambda_{k}$ and $\ket{\xi_{k}}$ are respectively the
eigenvalues and the eigenvectors of $\rho$.
Importantly, as emphasised by our notation in the definition \eref{eq:QFIUnit},
owing to the unitary parameter-encoding, the QFI is
independent of the estimated parameter and thus becomes just a function of
the state and the Hamiltonian.

\begin{thm}
\label{thm:ciaglosc}
For any pair of density matrices $\rho$ and $\sigma$
acting on $\mathcal{H}_{S}$ and for the QFI given in \eqnref{eq:QFIUnit}
the following inequalities hold true:
\begin{equation}
\left|F_\t{Q}[\rho;H]-F_\t{Q}[\sigma;H]\right|
\leq
32\,\sqrt{1-F^{2}(\rho,\sigma)}\,\|H\|^{2},\label{ciaglosc1}
\end{equation}
\begin{equation}
\left|F_\t{Q}[\rho;H]-F_\t{Q}[\sigma;H]\right|
\leq
32\, D_{B}(\rho,\sigma)\, \|H\|^{2},\label{ciaglosc2}
\end{equation}
and
\begin{equation}
\left|F_\t{Q}[\rho;H]-F_\t{Q}[\sigma;H]\right|
\leq
32\,\sqrt{\|\rho-\sigma\|_{1}}\,\|H\|^{2},\label{ciaglosc3}
\end{equation}
where $F$ and $D_{B}$ stand for the Uhlmann fidelity and the Bures distance
respectively.
\end{thm}

\begin{proof}The key ingredient of our proof is the fact that the QFI can
generally (\emph{not only} for unitary encodings) be expressed as
\begin{equation}
\left.F_\t{Q}[\rho_{\varphi}]\right|_{\varphi=\varphi_{0}}=4\min_{\ket{\psi_{\varphi}}}\langle\dot{\psi_{\varphi}}|\dot{\psi_{\varphi}}\rangle\Big|_{\varphi=\varphi_{0}},
\label{eq:QFI_purif}
\end{equation}
where $\ket{\dot{\psi_{\varphi}}}=\mathrm{d}\ket{\psi_{\varphi}}/\mathrm{d}\varphi$
and the minimization is in principle performed over all purifications
$\ket{\psi_{\varphi}}\in\mathcal{H}_{S}\ot\mathcal{H}_{E}$ of $\rho_{\varphi}$
for a given parameter true value $\varphi_0$
\cite{Fujiwara2008,Kolodynski2013,*Kolodynski2014}.

It turns out, however, that in this minimization it is enough to
consider only a family of purifications of $\rho_{\varphi}$ valid at $\varphi_0$ given by
\begin{equation}
\ket{\psi_{\varphi}}=\mathrm{e}^{-\mathrm{i}h_{E}(\varphi-\varphi_{0})}\ket{\widetilde{\psi}_{\varphi}},
\label{eq:gen_purif}
\end{equation}
where $\ket{\widetilde{\psi}_{\varphi}}$ is some fixed purification
of $\rho_{\varphi}$ and $h_{E}$ is any Hermitian operator acting
on the ancillary subsystem $\mathcal{H}_{E}$ (notice that $h_{E}$
is independent of $\varphi$) \cite{Kolodynski2013,*Kolodynski2014}. Moreover,
in our case, i.e., when the quantum evolution encoding the parameter
$\varphi$ is unitary, any purification of $\rho_{\varphi}$ takes
the form
\begin{equation}
\ket{\widetilde{\psi}_{\varphi}}=U_{\varphi}\otimes\mathbbm{1}_E\,\ket{\psi},
\label{eq:purif_rho_phi}
\end{equation}
for some purification $\ket{\psi}$ of $\rho$.

Now, by substituting \eqnsref{eq:gen_purif}{eq:purif_rho_phi} into
\eqnref{eq:QFI_purif} one obtains an equivalent expression for QFI given
by
\begin{equation}
F_\t{Q}[\rho;H]=4\min_{h_{E}}\langle\psi|(H+h_{E})^{2}|\psi\rangle,
\label{eq:QFI_purif_unit}
\end{equation}
in which:~$H+h_{E}=H\ot\mathbbm{1}_{E}+\mathbbm{1}_{S}\ot h_{E}$,
$H$ is the parameter-encoding Hamiltonian considered (acting on the system), and the
minimization is performed over all Hermitian operators $h_{E}$ acting on the environment.
It should be noticed that, in agreement with
definition \eref{eq:QFIUnit}, formula \eref{eq:QFI_purif_unit} no longer depends
on the parameter $\varphi$.

What is more, having the purification-based QFI
definition \eref{eq:QFI_purif_unit} for a unitary encoding at hand, we may explicitly
construct the optimal $h_E$ for a given Hamiltonian $H$ and a state $\rho=\sum_i\lambda_i\ket{\xi_i}\!\bra{\xi_i}$.
In particular, we may assume that the fixed purification of $\rho$ appearing
in \eqnref{eq:QFI_purif_unit} is the canonical one generated by the eigensystem of
$\rho$, that is,
\begin{equation}
\ket{\psi}=\sum_{i}\sqrt{\lambda_{i}}\ket{\xi_{i}}\ket{i}.
\end{equation}
Moreover, denoting by $h_{ij}$ the entries of $h_{E}$ in the standard basis of $\mathcal{H}_{E}$,
let us define the following multivariable function
\begin{eqnarray}
f(\{h_{ij}\}) & = & \langle\psi|(H+h_{E})^{2}|\psi\rangle\nonumber \\
 & = & \sum_{i}\lambda_{i}\langle\xi_{i}|H^{2}|\xi_{i}\rangle+\sum_{ij}\lambda_{i}h_{ij}h_{ji}\nonumber \\
 &  & +\;2\sum_{ij}\sqrt{\lambda_{i}\lambda_{j}}\langle\xi_{i}|H|\xi_{j}\rangle h_{ij}.
\end{eqnarray}
The necessary condition that this function has a minimum at some $h_E$
is that its derivatives over all $h_{ij}$ vanish at $h_E$. This gives
us the following system of equations
\begin{equation}
2\sqrt{\lambda_{i}\lambda_{j}}\langle\xi_{i}|H|\xi_{j}\rangle+(\lambda_{i}+\lambda_{j})h_{ji}=0
\end{equation}
 which implies that
\begin{equation}
h_{ij}=-\frac{2\sqrt{\lambda_{i}\lambda_{j}}}{{\lambda_{i}+\lambda_{j}}}\langle\xi_{j}|H|\xi_{i}\rangle.
\label{eq:hE_opt}
\end{equation}
The resulting matrix $h$ is clearly Hermitian. Moreover, due to the fact that the function $f$ is convex
[which follows from convexity of the square function (\cite[p. 113]{Bhatia2013})], the above solution
corresponds to its global minimum.

Now, stemming from the QFI definition \eref{eq:QFI_purif_unit} and the optimal
form of $h_E$ \eref{eq:hE_opt},
we prove the first QFI continuity relation \eref{ciaglosc1}.
Firstly, let $\ket{\psi_{\sigma}}$ and $h_{E}^{\sigma}$
be the purification of a given state $\sigma$ and the corresponding Hamiltonian
realising the minimum in \eqnref{eq:QFI_purif_unit} for this state. Furthermore,
let $\ket{\psi_{\rho}}$ be some, for the time being unspecified, purification
of another state $\rho$. At this point, it should be noticed that both purifications
$\ket{\psi_{\sigma}}$ and $\ket{\psi_{\rho}}$ can be chosen so that
they belong to the same Hilbert space (in other words, the ancillary
Hilbert space $\mathcal{H}_{E}$ can be taken the same for both purifications).

Let us finally assume, without any loss of generality, that $F_\t{Q}[\rho;H] \ge F_\t{Q}[\sigma;H]$,
i.e., $\rho$ is a better state with respect to the metrological task considered.
Noting that
\begin{eqnarray}
F_\t{Q}[\rho;H] & = & 4\min_{h_{E}}\langle\psi_{\rho}|(H+h_{E})^{2}|\psi_{\rho}\rangle\nonumber \\
 & \leq & 4\langle\psi_{\rho}|(H+h_{E}^{\sigma})^{2}|\psi_{\rho}\rangle,
\end{eqnarray}
we may upper-bound the QFI difference as
\begin{eqnarray}
F_\t{Q}[\rho;H]-F_\t{Q}[\sigma;H] & \leq & 4\left[\langle\psi_{\rho}|(H+h_{E}^{\sigma})^{2}|\psi_{\rho}\rangle\right.\\
 &  & \hspace{0.4cm}\left.-\langle\psi_{\sigma}|(H+h_{E}^{\sigma})^{2}|\psi_{\sigma}\rangle\right]\nonumber \\
 & = & 4\Tr[(\psi_{\rho}-\psi_{\sigma})(H+h_{E}^{\sigma})^{2}],\nonumber
\end{eqnarray}
where by $\psi_{\rho}$ and $\psi_{\sigma}$ we denote projectors
onto $\ket{\psi_{\rho}}$ and $\ket{\psi_{\sigma}}$ respectively.
Moreover, exploiting the fact that
\begin{equation}
|\Tr(A^{\dagger}B)|\leq\|A\|_{1}\|B\|
\end{equation}
holds for any two operators $A$ and $B$ with $\|\cdot\|$ denoting the matrix norm $\|X\|:=\max_{\|\psi\|=1}\|X|\psi\rangle\|$,
we arrive at the following expression:
\begin{equation}
F_\t{Q}[\rho;H]-F_\t{Q}[\sigma;H]\leq4\,\|\psi_{\rho}-\psi_{\sigma}\|_{1}\,\|H+h_{E}^{\sigma}\|^{2}.
\label{eq:diffbound_TN}
\end{equation}

To obtain a similar relation for the states $\rho$ and $\sigma$
instead of their purifications, let us notice that for any two pure
states $\ket{\psi}$ and $\ket{\phi}$:
\begin{equation}
\|\psi-\phi\|_{1}=2\sqrt{1-F^{2}(\ket{\psi},\ket{\phi})},
\label{TNFidelity}
\end{equation}
where $F$ stands for the Uhlmann fidelity (\ref{eq:fidelity}).
Thus, we may rewrite \eqnref{eq:diffbound_TN} to obtain
\begin{equation}
F_\t{Q}[\rho;H]-F_\t{Q}[\sigma;H]
\leq
8
\sqrt{1-F^{2}(\ket{\psi_{\rho}},\ket{\psi_{\sigma}})}
\,\|H+h_{E}^{\sigma}\|^{2}.
\label{eq:diffbound_fid}
\end{equation}
Crucially, we can still exploit the freedom in choosing
the purification of the state $\rho$. Concretely, we can choose it to be
the one that realises maximum in \eqnref{eq:fidelityPur}, which allows
us to just write $F(\ket{\psi_{\rho}},\ket{\psi_{\sigma}})=F(\rho,\sigma)$.
Hence, \eqnref{eq:diffbound_fid} rewrites as
\begin{equation}
F_\t{Q}[\rho;H]-F_\t{Q}[\sigma;H]\leq
8
\sqrt{1-F^{2}(\rho,\sigma)}
\|H+h_{E}^{\sigma}\|^{2}.
\label{HTC}
\end{equation}
In order to turn the above inequality into the one of \eqnref{ciaglosc1},
we need to make \eqnref{HTC} independent of the
the auxiliary Hamiltonian $h_{E}^{\sigma}$.
We achieve this by proving that its norm can always be upper-bounded
by the norm of the parameter-encoding Hamiltonian, i.e., $\|h_{E}^{\sigma}\|\leq\|H\|$.

For this purpose, we recall that $h_{E}^{\sigma}$ that realises the minimum in
\eqnref{eq:QFI_purif_unit} for the state $\sigma$
(with eigendecomposition $\sigma\!=\!\sum_i\mu_i\ket{\eta_i}\!\bra{\eta_i}$)
must have the form derived in \eqnref{eq:hE_opt}.
Then, we note that for a Hermitian operator $h$ its operator norm can be expressed as
\begin{equation}
\|h\|:=\max_{\ket{\psi},\|\psi\|=1}|\langle\psi|h|\psi\rangle|.
\end{equation}
Let then $\ket{\omega}$ denote the pure state realizing the above maximum for $h_E^{\sigma}$
(it is just the eigenvector of $h_E^{\sigma}$
corresponding to its eigenvalue with the largest absolute
value). Writing $\ket{\omega}$ in the standard basis as
$\ket{\omega}=\sum_{i}\alpha_{i}\ket{i}$, it follows from \eqnref{eq:hE_opt} that the operator
norm of $h_{E}^{\sigma}$ is thus given by
\begin{equation}
\|h_{E}^{\sigma}\|=\left|\sum_{ij}\alpha_{i}^{*}\alpha_{j}\frac{2\sqrt{\mu_{i}\mu_{j}}}{\mu_{i}+\mu_{j}}\langle\eta_{j}|H|\eta_{i}\rangle\right|.\label{Faceless}
\end{equation}
Now, we introduce the following vector
\begin{equation}
\ket{\eta(t)}=\sum_{i}\alpha_{i}^{*}\sqrt{\mu_{i}}\mathrm{e}^{-t\mu_{i}}\ket{\eta_{i}}
\end{equation}
with $t\in[0,\infty)$ being some parameter.
This vector is normalised so that $\int_{0}^{\infty}\mathrm{d}t\,\langle\eta(t)|\eta(t)\rangle=1/2$.
As a result, we may write the operator norm of $h_{E}^{\sigma}$ as follows
\begin{eqnarray}
\|h_{E}^{\sigma}\| & = & 2\left|\int_{0}^{\infty}\mathrm{d}t\,\langle\eta(t)|H|\eta(t)\rangle\right|
\end{eqnarray}
and, realising that
\begin{eqnarray}
\left|\int_{0}^{\infty}\mathrm{d}t\,\langle\eta(t)|H|\eta(t)\rangle\right| & \leq & \|H\|\left|\int_{0}^{\infty}\mathrm{d}t\,\langle\eta(t)|\eta(t)\rangle\right|\nonumber \\
 & = & \frac{1}{2}\|H\|,
\end{eqnarray}
we prove that indeed $\|h_{E}^{\sigma}\|\leq\|H\|$. As a result, we may upper-bound the norm appearing in \eqnref{HTC} as
\begin{equation}
\|H+h_{E}\|^{2}\leq(\|H\|+\|h_{E}\|)^{2}\leq4\|H\|^{2}
\end{equation}
and finally arrive at the the first continuity relation \eref{ciaglosc1}.

To prove the second continuity relation \eref{ciaglosc2},  it is enough
to notice that
\begin{eqnarray}
\sqrt{1-F^{2}(\rho,\sigma)} & = & \sqrt{[1-F(\rho,\sigma)][1+F(\rho,\sigma)]}\\
 & \leq & \sqrt{2\left[1-F(\rho,\sigma)\right]}=D_{B}(\rho,\sigma),\nonumber
\end{eqnarray}
where to obtain the inequality we have used the fact that $F(\rho,\sigma)\leq1$
for any pair of states $\rho,\sigma$.

Lastly, in order to prove the third continuity relation \eref{ciaglosc3}, we exploit
the Fuchs--van de Graaf inequality \cite{Bengtsson2006}, which states that
for any pair of density matrices $\rho$ and $\sigma$ acting on $\mathbbm{C}^{d}$
\begin{equation}
1-F(\rho,\sigma)\leq\frac{1}{2}\|\rho-\sigma\|_{1}.
\end{equation}
Hence, it directly follows that $D_{B}(\rho,\sigma)\leq\sqrt{\|\rho-\sigma\|_{1}}$
and we obtain the last inequality of \eqnref{ciaglosc3}.
\end{proof}

We now consider the case where the two states $\rho$
and $\sigma$ are pure. In this situation, the continuity relations
of Theorem~\ref{thm:ciaglosc}---in particular \eqnref{ciaglosc1}---can be improved
by a factor of $3/4$, as demonstrated in the following lemma.

\begin{lem}
For any pair of pure states~$\ket{\psi},\ket{\phi}\in\mathcal{H}_{S}$, and a parameter-encoding Hamiltonian $H$, the
following inequality holds
\begin{eqnarray}
\left|F_\t{Q}[\ket{\psi};H]-F_\t{Q}[\ket{\phi};H]\right| & \leq & 12\|\psi-\phi\|_{1}\|H\|^{2}\nonumber \\
 & = & 24\sqrt{1-F^{2}(\ket{\psi},\ket{\phi})}\,\|H\|^{2},\nonumber \\
\label{eq:ciaglosc_pure}
\end{eqnarray}
where $\psi$ and $\phi$ denote the projectors onto $\ket{\psi}$ and
$\ket{\phi}$, respectively.
\end{lem}

\begin{proof}
We begin by recalling that the QFI for pure states reads
\begin{equation}
F_\t{Q}[\ket{\psi};H]=4(\langle\psi|H^{2}|\psi\rangle-\langle\psi|H|\psi\rangle^{2}).
\end{equation}
This allows us to upper-bound the left-hand side of \eqnref{eq:ciaglosc_pure} as
\begin{eqnarray}
\left|F_\t{Q}[\ket{\psi};H]-F_\t{Q}[\ket{\phi};H]\right| & \leq & 4|\Tr(H^{2}\psi)-\Tr(H^{2}\phi)|\nonumber \\
 &  & +4|[\Tr(H\psi)]^{2}-[\Tr(H\phi)]^{2}|.\nonumber \\
\label{nierownosc}
\end{eqnarray}
Let us now concentrate on the second term appearing on the right-hand side of the above inequality.
It can be bounded from above as
\begin{eqnarray}
\left|[\Tr(H\psi)]^{2}-[\Tr(H\phi)]^{2}\right| & = & |\Tr(H\psi)-\Tr(H\phi)|\nonumber \\
 &  & \times|\Tr(H\psi)+\Tr(H\phi)|\nonumber \\
 & \leq & 2|\Tr(H\psi)-\Tr(H\phi)|\|H\|,\nonumber \\
\label{eq:remik_to_pala}
\end{eqnarray}
where the last inequality is a consequence of the fact that for any normalized
$\ket{\psi}$:~$\Tr(\psi H)=\langle\psi|H|\psi\rangle\leq\|H\|$.
Plugging \eqnref{eq:remik_to_pala} into \eqnref{nierownosc}, we obtain
\begin{eqnarray}
\left|F_\t{Q}[\ket{\psi};H]-F_\t{Q}[\ket{\phi};H]\right| & \leq & 4\left|\Tr[(\psi-\phi)H^{2}]\right|\nonumber \\
 &  & +8|\Tr[(\psi-\phi)H]|\|H\|\nonumber \\
\end{eqnarray}
and finally, acknowledging that $|\Tr(X^{\dagger}Y)|\leq\|X\|_{1}\|Y\|$
holds for any two operators $X,Y$, we arrive at
\begin{eqnarray}
\left|F_\t{Q}[\ket{\psi};H]-F_\t{Q}[\ket{\phi};H]\right| & \leq & 12\|\psi-\phi\|_{1}\|H\|^{2}\nonumber \\
 & = & 24\sqrt{1-F^{2}(\ket{\psi},\ket{\phi})}\,\|H\|^{2},\nonumber \\
\end{eqnarray}
where the last equality stems from \eqnref{TNFidelity}.
\end{proof}

Although at first sight the inequality \eref{eq:ciaglosc_pure}
may seem to be less important due the constraint of states purity,
it crucially allows us also to tighten the QFI continuity relation \eref{ciaglosc1}
for the case when one of the states considered is pure, while the other
can possibly be mixed. In fact, we are able to do so by using the
\emph{convex-roof--based} definition of the QFI that has been introduced for protocols with
unitary encoding in \refcite{Toth2013,*Yu2013}:

\begin{thm}
\label{thm:ciaglosc_pure-mixed}
For any mixed state $\rho$ acting on $\mathcal{H}_S$ and any pure state $\ket{\phi}\in\mathcal{H}_S$
and a given parameter-encoding Hamiltonian $H$,
the following inequality holds true:
\begin{equation}
\left|F_\t{Q}[\rho;H]-F_\t{Q}[\ket{\phi};H]\right|\leq24\sqrt{1-F^{2}(\rho,\ket{\phi})}\,\|H\|^{2}.
\label{eq:ciaglosc_pure-mixed}
\end{equation}
\end{thm}

\begin{proof}
Let us first recall that the QFI of a mixed
state can be expressed as the convex roof of the variance \cite{Toth2013,Yu2013}, i.e.,
\begin{equation}
F_\t{Q}[\rho;H]=\inf_{\{\lambda_{k},\ket{\xi_{k}}\}}\sum_{k}p_{k}\,F_\t{Q}[\ket{\xi_{k}};H],
\label{roof}
\end{equation}
where the infimum is taken over all ensembles $\{p_{k},\ket{\xi_{k}}\}$
such that $\sum_{k}p_{k}\proj{\xi_{k}}=\rho$ (importantly $\ket{\xi_k}$ are normalised but
generally \emph{not} orthogonal).
Choosing then $\{p_{k},\ket{\xi_{k}}\}$
to be the ensemble realising the minimum in \eqnref{roof}, one finds that
\begin{eqnarray}
\left|F_\t{Q}[\rho;H]\!-\!F_\t{Q}[\ket{\phi};H]\right| & \leq & \sum_{k}p_{k}|F_\t{Q}[\ket{\xi_{k}};H]-F_\t{Q}[\ket{\phi};H]|\nonumber \\
 & \leq & 24\sqrt{1\!-\!\sum_{k}p_{k}F^{2}(\ket{\xi_{k}},\ket{\phi})}\,\|H\|^{2}\nonumber \\
 & = & 24\sqrt{1-F^{2}(\rho,\ket{\phi})}\,\|H\|^{2},
\end{eqnarray}
where the second inequality follows from the pure-states continuity relation \eref{eq:ciaglosc_pure}
and the concavity of the square root.
\end{proof}

Lastly, we adopt the above proved Theorems \ref{thm:ciaglosc} and \ref{thm:ciaglosc_pure-mixed} to
the case of the metrology protocol considered in the main text, i.e., the setting when the system investigated
consists of $N$ particles, each independently sensing a unitarily encoded parameter of interest.
Then, we may always express the overall system Hamiltonian $H$ as a sum of the local ones:
\begin{equation}
H_\mathrm{loc}:=\sum_{n=1}^{N}h^{(n)},
\label{eq:H_N}
\end{equation}
where $h^{(n)}$ represents the parameter-encoding Hamiltonian
of the $n$-th particle and is conveniently normalised so that $\|h^{(n)}\|\leq1/2$ for all $n$.
In what follows we refer to such Hamiltonians as \textit{local} and denote them by $H_{\mathrm{loc}}$.
In particular, as for any $H_{\mathrm{loc}}$, $\|H_{\mathrm{loc}}\|^2\le N^2/4$,
the three general QFI continuity relations \eref{ciaglosc1}--\eref{ciaglosc3} yield the following:

\begin{cor}
For any pair of $N$-particle states $\rho^N$ and $\sigma^N$
acting on $(\mathbbm{C}^d)^{\ot N}$, and any local Hamiltonian $H_{\mathrm{loc}}$,
the difference in the QFIs of $\rho^N$ and $\sigma^N$
can always be upper-bounded as:

\begin{equation}
|F_\t{Q}[\rho^N;H_{\mathrm{loc}}]-F_\t{Q}[\sigma^N;H_{\mathrm{loc}}]|
\leq
8
\,\sqrt{1-F^{2}(\rho^N,\sigma^N)}\,N^{2},
\label{ciaglosc1_N}
\end{equation}
\begin{equation}
|F_\t{Q}[\rho^N;H_{\mathrm{loc}}]-F_\t{Q}[\sigma^N;H_{\mathrm{loc}}]|\leq8\,D_{B}(\rho^N,\sigma^N)\,N^{2}
\end{equation}
and
\begin{equation}
|F_\t{Q}[\rho^N;H_{\mathrm{loc}}]-F_\t{Q}[\sigma^N;H_{\mathrm{loc}}]|\leq8\,\sqrt{\|\rho^N-\sigma^N\|_{1}}\, N^{2}.
\end{equation}
\end{cor}
Analogously, we may then also rewrite Theorem~\ref{thm:ciaglosc_pure-mixed}, which
deals with the case of the one of the states being pure.
\begin{cor}
For any pair of $N$-particle states, a mixed $\rho^N$ and a pure $\ket{\phi^N}$,
and a local Hamiltonian $H_{\mathrm{loc}}$,
the continuity relation \eref{eq:ciaglosc_pure-mixed} leads to
\begin{equation}
\!\!\left|F_\t{Q}[\rho^N;H_{\mathrm{loc}}]\!-\!F_\t{Q}[\ket{\phi^N};H_{\mathrm{loc}}]\right|
\!\leq
6\sqrt{1\!-\!F^{2}(\rho^N,\ket{\phi^N})}\,N^{2}.
\label{eq:ciaglosc_pure-mixed_N}
\end{equation}
\end{cor}

\section{Relating the QFI to geometric measures of entanglement}
\label{app:geom_meas}

Here we show that the continuity relation (\ref{ciaglosc1}) can be
used to link the QFI to a \emph{multiparite entanglement measure}.
To this end, we first need to recall the notion of $k$-producibility.

Consider a multipartite pure state $\ket{\psi^{N}}\in(\mathbbm{C}^d)^{\ot N}$.
We call it \textit{$k$-producible} with $k\leq N$ if it can
be written as \citep{Guhne2005,Guhne2006}
\begin{equation}
\ket{\psi^{N}}=\ket{\psi_{1}}\ot\ldots\ot\ket{\psi_{m}},
\end{equation}
with each $\ket{\psi_{i}}$ being a pure state consisting of \emph{at
most} $k$ parties. In particular, it follows from this definition
that a $k$- but not $(k-1)$-producible state contains $k$ particles
that are genuinely entangled \cite{Horodecki2009,Guhne2009}.

This definition can be straightforwardly extended
to mixed states: a mixed state $\rho^{N}$ is $k$-producible
if it is a probabilistic mixture of $k$-producible pure states. By
definition, for every $k$, the set of all $k$-producible states
$\mathcal{S}^N_{k}$ is convex. 
Moreover, such sets of $k$-producible states form a hierarchy
that we schematically depict in \figref{fig:sets}.
In particular, $\mathcal{S}^N_{1}$ contains all fully separable states 
and $\mathcal{S}^N_{N}$ is the set of all states, and, in general, 
$\mathcal{S}^N_{1}\subset\ldots\subset\mathcal{S}^N_{N}$.
Note that the set $\mathcal{S}^N_N\setminus\mathcal{S}^N_{N-1}$
thus consists of all $N$-partite genuinely entangled states.


Exploiting the fact that the sets $\mathcal{S}^N_{k}$ are convex, one
can easily introduce entanglement measures quantifying the extent
to which a given $N$-partite state is non-$k$-producible. More concretely,
for pure states one defines
\begin{equation}\label{wielkadupa}
E_{k}^{\mathrm{prod}}[\ket{\psi^{N}}]:=1-\max_{\ket{\phi^{N}}\in
\mathcal{S}^N_{k}}|\langle\phi^{N}|\psi^{N}\rangle|^{2},
\end{equation}
which is then extended to mixed states by using the convex roof construction,
i.e.,
\begin{equation}
E_{k}^{\mathrm{prod}}[\rho^{N}]:=\inf_{\{p_{i},\ket{\psi_{i}^{N}}\}}\sum_{i}p_{i}E_{k}^{\mathrm{prod}}[\ket{\psi_{i}^{N}}].
\label{KprodMixed}
\end{equation}
The infimum above is taken over all ensembles $\{p_{i},\ket{\psi_{i}^{N}}\}$
realising $\rho^{N}$, i.e., such that $\sum_{i}p_{i}\proj{\psi_{i}^{N}}=\rho^{N}$
(importantly $\ket{\psi_{i}^{N}}$ are normalised but generally \emph{not} orthogonal).

It is important to note that \eqnref{KprodMixed} can be rewritten
with help of Uhlmann fidelity \eref{eq:fidelity}, (see the appendices of \refcite{Streltsov2010}),
so that the optimisation can be performed over all k-producible mixed states:
\begin{equation}
E_{k}^{\mathrm{prod}}[\rho^{N}]=1-\max_{\sigma^{N}\in\mathcal{S}^N_{k}}F^{2}(\rho^{N},\sigma^{N}),
\label{producibility}
\end{equation}
what allows us to relate $E_k^\t{prod}$ to the QFI.

Lastly, let us mention that for the special case of $k\!=\!1$
in \eqnref{KprodMixed}, one recovers the definition of the
\emph{geometric measure of entanglement} (GME), $E_\t{G}[\rho^N]=E_1^\t{prod}[\rho^N]$, that we only consider
in the main text of this work.

\begin{lem}
For any pure $\ket{\psi^{N}}\in(\mathbbm{C}^d)^{\ot N}$
and any local Hamiltonian $H_{\mathrm{loc}}$,
the following inequality holds:
\begin{equation}
F_{Q}[\ket{\psi^{N}};H_{\mathrm{loc}}]\leq kN+6\sqrt{E_{k}^{\mathrm{prod}}(\ket{\psi^{N}})}\,N^{2}.\label{EntQFI2Pure}
\end{equation}
\end{lem} \begin{proof}

Denoting by $\ket{\phi_{*}^{N}}$ the $k$-producible state realizing the maximum in
Eq. (\ref{wielkadupa}) for $\ket{\psi^N}$, i.e.,

\begin{equation}
E_G[\ket{\psi^N}]=1-
|\langle\phi_{*}^{N}|\psi^{N}\rangle|^{2},
\end{equation}
it follows from \eqnref{eq:ciaglosc_pure-mixed_N} that
\begin{equation}
F_{Q}[\ket{\psi^{N}}]\leq F_{Q}[\ket{\phi_{*}^{N}}]+6\sqrt{E_{k}^{\mathrm{prod}}[\ket{\psi^{N}}]}\,N^{2}.
\end{equation}

In order to obtain \eqnref{EntQFI2Pure} and complete the proof, it remains to
utilise the fact that for any $k$-producible state $\sigma^{N}\in\mathcal{S}_k^N$,
its QFI is upper-bounded as follows \citep{Toth2012,Hyllus2012}:
\begin{equation}
F_{Q}[\sigma^{N}]\le\left\lfloor \frac{N}{k}\right\rfloor k^{2}+\left(N-\left\lfloor \frac{N}{k}\right\rfloor k\right)^{2}\le kN.
\end{equation}
\end{proof}

Exploiting the above lemma, we can now prove the following general
theorem.

\begin{thm}\label{EntQFI} For any state $\rho^{N}$ acting on $(\mathbbm{C}^{d})^{\otimes N}$
and any local Hamiltonian $H_{\mathrm{loc}}$,
the following inequality is true:
\begin{equation}
F_{Q}[\rho^{N};H_{\mathrm{loc}}]\leq kN+6\sqrt{E_{k}^{\mathrm{prod}}[\rho^{N}]}\,N^{2},\label{EntQFI1}
\end{equation}

\end{thm} \begin{proof}Let $\{p_{i},\ket{\psi_{i}^{N}}\}$ be an
ensemble realising $\rho^{N}$ for which the minimum in \eqnref{KprodMixed}
is achieved. Then, we have the following chain of inequalities
\begin{eqnarray}
F_{Q}[\rho^{N};H] & \leq & \sum_{i}p_{i}F_{Q}[\ket{\psi_{i}^{N}};H]\nonumber \\
 & \leq & kN+6\sum_{i}p_{i}\sqrt{E_{k}^{\mathrm{prod}}[\ket{\psi_{i}^{N}}]}\,N^2\nonumber \\
 & \leq & kN+6\sqrt{\sum_{i}p_{i}E_{k}^{\mathrm{prod}}[\ket{\psi_{i}^{N}}]}\,N^2\nonumber \\
 & = & kN+6\sqrt{E_{k}^{\mathrm{prod}}[\rho^{N}]}\,N^2,\nonumber \\
\end{eqnarray}
where the second and the third inequalities follow respectively from \eqnref{EntQFI2Pure}
and the concavity of the square root, while the last equality
stems from the definition of $E_{k}^{\mathrm{prod}}$ \eref{KprodMixed}.
\end{proof}

\textit{Remark.} For $k=1$, inequality \eref{EntQFI1} relates
the QFI for any Hamiltonian of the form \eref{eq:H_N}, $H_\mathrm{loc}$,
to the \emph{geometric measure of entanglement} $E_\t{G}$
used in the main text:
\begin{equation}
F_{Q}[\rho^{N};H_\mathrm{loc}]\leq N+6\sqrt{E_\t{G}[\rho^{N}]}\,N^{2}.
\end{equation}
On the other hand, \eref{EntQFI1} can be used to derive a lower
bound on $E_{k}^{\mathrm{prod}}$:
\begin{equation}
E_{k}^{\mathrm{prod}}(\rho^{N})\geq
\left\{
\begin{array}{cc}
\left(\frac{F_{Q}[\rho^{N};H_\mathrm{loc}]-kN}{6N^{2}}\right)^{2}, & F_{Q}[\rho^{N};H_\mathrm{loc}]> kN\\
0, & F_{Q}[\rho^{N};H_\mathrm{loc}]\leq kN
\end{array}
\right.
\label{bound}
\end{equation}
whose right-hand side scales with $N$ as $(F_{Q}[\rho^{N};H_\mathrm{loc}]/6N^{2})^2$
in the limit of large $N$.

The bound (\ref{bound}) is in general not tight.
For instance, for the $N$-qubit GHZ state
\begin{equation}\label{GHZ}
\ket{\psi^N_\textrm{\tiny GHZ}}=\frac{1}{\sqrt{2}}(\ket{0}^{\ot N}+\ket{1}^{\ot N})
\end{equation}
the QFI with the Hamiltonian $H_\mathrm{loc}=(1/2)\sum_{i}\sigma_{i}^{z}$ amounts
to $F_{Q}[\ket{\psi^N_\textrm{\tiny GHZ}}]=N^{2}$, and hence our bound gives
$E_{\t{G}}[\ket{\psi^N_\textrm{\tiny GHZ}}]\geq(N^{2}-N)^{2}/36N^{4}$,
which tends to $1/36$ for $N\to\infty$, while it is known that $E_{\t{G}}[\ket{\psi^N_\textrm{\tiny GHZ}}]=1/2$.
Nevertheless, it allows one to lower bound $E_{k}^{\mathrm{prod}}$
for states for which only the QFI is easy to compute. Generally speaking,
the bound \eref{bound} provides a non-trivial estimation of $E_{k}^{\mathrm{prod}}$
for all states for which $F_{Q}[\rho^{N};H_\mathrm{loc}]>kN$.

\section{Estimating GME for Werner-type states}
\label{app:EG_est}

Let us consider the following class of $N$-qubit Werner-type states, i.e., a mixture
of the GHZ state (\ref{GHZ}) and the maximally mixed state:

\begin{equation}
\rho^N_p=p\proj{\psi^N_\textrm{\tiny GHZ}}+(1-p)\frac{\mathbbm{1}_{2^N}}{2^N}.\label{GHZwithWhitenoise}
\end{equation}
The GME for these states
can be upper bounded as $E_\t{G}[\rho^N_p]\leq p/2$. This follows from the facts that
$E_\t{G}$ is convex and that $E_\t{G}[\ket{\psi^N_\textrm{\tiny GHZ}}]=1/2$ for any $N$. Our aim here is to show
that for sufficiently large $N$ this upper bound is very close to the value of
$E_\t{G}[\rho^N_p]$.

To this end, let us first notice that very recently in
\refcite{GHZStates} it has been shown that computation of
$E_\t{G}[\rho^N_p]$ simplifies to the following maximization
\begin{equation}
E_\t{G}[\rho^N_p]=\max_{\mu\in[0,\mu_m]}f^N_p(\mu),
\end{equation}
where $\mu_m=2^{N-3}/(2^{N-2}-1)$ and
\begin{eqnarray}
f^N_p(\mu)&=&\frac{1}{2}\left[1-\mu-\sqrt{\gamma}+2p\mu\right.\nonumber\\
&&\hspace{0.5cm}\left.+\frac{1-p}{2^N}\left(2\mu+\frac{\mu(\mu+\sqrt{\alpha})}{\mu-1}\right)\right]
\end{eqnarray}
with $\gamma=(\mu-1)^2+2^{3-N}\mu$ and $\alpha=1-\mu+\mu^2$.

Now, it is clear that $E_\t{G}[\rho^N_p]\geq f_p^N(\mu_m)$. It is also not difficult to see that
for $N\to\infty$, $\mu_m\to 1/2$, $\gamma_m\to 1/4$ and $\alpha_m\to 3/4$, where $\gamma_m$ and $\alpha_m$
are $\gamma$ and $\alpha$ computed for $\mu_m$. All this implies that
$f_p^N(\mu_m)\to p/2$, and thus $E_\t{G}[\rho^N_p]\to p/2$ for large $N$.

\begin{figure}[!t]
\includegraphics[width=1\columnwidth]{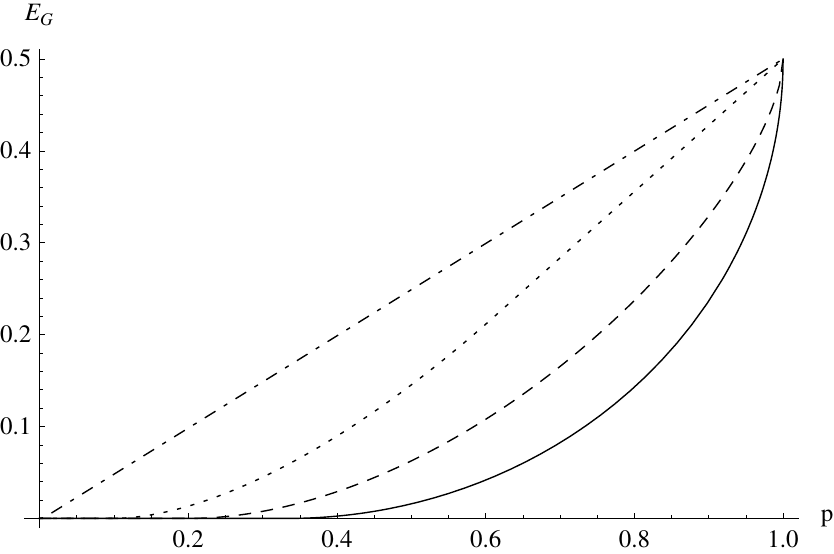}
\caption{%
\textbf{Geometric measure of entanglement} (GME), $E_{G}$, as a function of the
parameter $p$ for states given in Eq.~(\ref{GHZwithWhitenoise}). The number
of qubits, $N$, is chosen to be:~2 (\emph{solid line}), 3 (\emph{dashed line}), 4 (\emph{dotted
line}), and 10 (\emph{dot-dashed line}).}
\label{fig:Eg}
\end{figure}

Furthermore, one should note that the convergence of $E_\t{G}[\rho^N_p]$ to $p/2$ with $N\to \infty$
is quite fast. In other words, already for systems of moderate size ($N=10$)
the upper bound $E_\t{G}[\rho^N_p]\leq p/2$ is a good approximation to $E_\t{G}[\rho^N_p]$.
For this purpose, let us consider the following rough estimation of
$|f_p^N(\mu_m)-p/2|$. We first notice that
\begin{eqnarray}\label{Raimat}
\left|f_p^N(\mu_m)-\frac{p}{2}\right|&\leq& \frac{1}{2}\left|1-\mu_m-\sqrt{\gamma_m}\right|+\frac{p}{2}|2\mu_m-1|\nonumber\\
&&+\frac{1-p}{2^N}\left|2\mu_m+\frac{\mu_m(\mu_m+\sqrt{\alpha_m})}{\mu_m-1}\right|\nonumber\\
\end{eqnarray}
Let us now bound each of the three terms appearing in the above expression. First, we see that
\begin{eqnarray}
|1-\mu_m-\sqrt{\gamma_m}|\leq \frac{1}{2^{N-2}-1}+\sqrt{\frac{1}{2^{N-2}-1}}
\end{eqnarray}
Second,
\begin{equation}
|2\mu_m-1|= \frac{1}{2^{N-2}-1},
\end{equation}
And finally, for $N\geq 4$,
\begin{equation}
\left|2\mu_m+\frac{\mu_m(\mu_m+\sqrt{\alpha_m})}{\mu_m-1}\right|\leq \beta
\end{equation}
with $\beta=4/3+2\left(2+\sqrt{7}\right)/3\approx 4.43$.
All this gives
\begin{equation}
\left|f_p^N(\mu_m)-\frac{p}{2}\right|\leq \frac{1}{2^{N-2}-1}+\frac{1}{2}\sqrt{\frac{1}{2^{N-2}-1}}+\frac{\beta}{2^N}.
\end{equation}
One then sees that already for $N=10$, the difference between then upper bound and the actual value of the
GME for $\rho^N_p$ is at most $0.04$.

To demonstrate the fast convergence, we have plotted in \figref{fig:Eg} the GME, 
$E_\t{G}$, as a function of the parameter $p$ for $N$
being:~2, 3, 4, and 10. For large $N$, the curve becomes almost indistinguishable from $p/2$,
which is clear on \figref{fig:Eg} already for $N=10$.

\bibliographystyle{apsrev4-1}
\bibliography{ent_metro}

\end{document}